\documentclass[11pt,onecolumn]{IEEEtran}

\usepackage{amsmath,amsthm,amssymb,amsfonts}
\usepackage{enumerate}
\usepackage{cite}
\usepackage{subfigure}
\usepackage{url}

\newtheorem{theorem}{Theorem}
\newtheorem{corollary}{Corollary}[theorem]
\newtheorem{lemma}[theorem]{Lemma}

\newtheorem{mydef}{Definition}

\usepackage{tikz-cd}
\usepackage{footnote}
\usepackage{hyperref}
\usepackage[utf8]{inputenc}

\usepackage{bbm}

\usepackage{algorithm}
\usepackage[noend]{algpseudocode}
\makeatletter
\def\BState{\State\hskip-\ALG@thistlm}
\makeatother

\begin{document}
\title{Arbitrarily Varying Networks: Capacity-achieving Computationally Efficient Codes}
\author{Peida Tian, Sidharth Jaggi, Mayank Bakshi, Oliver Kosut
\thanks{P. Tian, S. Jaggi, and M. Bakshi are with the Chinese University of Hong Kong (e-mail: \hbox{tianpeida.cuhk@gmail.com}, \hbox{jaggi@ie.cuhk.edu.hk}, \hbox{mayank@inc.cuhk.edu.hk}). The work of P. Tian, S. Jaggi, and M. Bakshi described in this paper was partially supported by a grant from University Grants Committee of the Hong Kong Special Administrative Region, China (Project No. AoE/E-02/08).}
\thanks{O. Kosut is with the School of Electrical, Computer and Energy Engineering, Arizona State University (e-mail: \hbox{okosut@asu.edu}). This material is based upon work supported by the National Science Foundation under Grant No. CCF-1422358.}
}
\maketitle

\begin{abstract}
\label{sec:abstract}
We consider the problem of communication over a network containing a hidden and malicious adversary that can control a subset of network resources, and aims to disrupt communications. We focus on omniscient node-based adversaries, \emph{i.e.}, the adversaries can control a subset of nodes, and know the message, network code and packets on all links. Characterizing information-theoretically optimal communication rates as a function of network parameters and bounds on the adversarially controlled network is in general open, even for \emph{unicast} (single source, single destination) problems. In this work we characterize the information-theoretically optimal \emph{randomized capacity} of such problems, \emph{i.e.}, under the assumption that the source node shares (an asymptotically negligible amount of) independent common randomness with each network node \emph{a priori} (for instance, as part of network design). We propose a novel computationally-efficient communication scheme whose rate matches a natural information-theoretically ``erasure outer bound" on the optimal rate. Our schemes require no prior knowledge of network topology, and can be implemented in a distributed manner as an overlay on top of classical distributed linear network coding.
\end{abstract}

\section{Introduction}
\label{sec:introduction}
Network coding allows routers in networks to mix packets. This helps attain information-theoretically throughput for a variety of network communication problems; in particular for {\it network multicast}~\cite{ahlswede2000network,ho2003benefits}, often via linear coding operations~\cite{li2003linear,koetter2003algebraic}. Throughput-optimal network codes can be efficiently designed~\cite{jaggi2005polynomial}, and may even be implemented distributedly~\cite{ho2006random}. Also, network-coded communication is more robust to packet losses/link-failures~\cite{koetter2003algebraic,ho2003benefits,lun2005efficient}. 

However, when the network contains malicious nodes/links, due to the mixing nature of network coding, even a single erroneous packet can cause all packets at the receivers being corrupted. This motivates the problem of network error correction, which was first studied by Cai and Yeung in~\cite{yeung2006network,cai2006network}. They considered an {\it omniscient} adversary capable of injecting errors on {\it any} $z$ links, and showed that $C-2z$ was both an inner and outer bound on the optimal throughput, where $C$ is the {\it network-multicast min-cut}. Jaggi {\it et al.}~\cite{jaggi2007resilient} proposed efficient network codes to achieve this rate. In parallel, K{$\ddot{\text{o}}$}tter and Kschischang~\cite{koetter2008coding} developed a different and elegant approach based on subspace/rank-metric codes to achieve the same rate. Furthermore, when the adversary is of ``limited-view" in some manner (for instance, adversary can observe only a sufficiently small subset of transmissions, or is computationally bounded, or is ``causal/cannot predict future transmissions"), a higher rate is possible, and in fact~\cite{jaggi2007resilient,yao2014network,nutman2008adversarial} proposed a suite of network codes that achieve $C-z$, all of which meet the network Hamming bound in~\cite{yeung2006network}. A more refined adversary model is considered in \cite{zhang2015coding}.

Although communication in the presence of link-based adversaries is now relatively well-understood, problems where the adversaries are ``node-based" seem to be much more challenging. In node-based case, the adversaries can attack any subset of at most $z$ nodes by injecting errors on outgoing edges of those nodes. Since the adversary is restricted to control nodes, this places restrictions on the subsets of links it can control. This problem was first studied by Kosut {\it et al.} in~\cite{kosut2010adversaries, kosut2014polytope}, where it is shown that reducing node-based adversary to link-based one is too coarse, and linear codes are insufficient in general. A class of non-linear network codes was proposed~\cite{kosut2014polytope} to achieve capacity for a subset of planer networks, but the general problem of characterizing network capacity with node-based adversaries is still open. 

This problem has been studied from various perspectives. A cut-set bound was given in \cite{kosut2010adversaries, kosut2014polytope}. The routing-only capacity was studied in~\cite{che2013routing}. The work of \cite{kim2011network} explored the unequal link capacities, and \cite{kosut2014generalized} considered a general problem formulation subsuming both link-based and node-based adversaries. The fundamental complexity was examined in \cite{huang2015connecting}, where the authors showed the general network error correction is as hard as a long standing open problem, {\it i.e.} multiple unicast network coding. 

On the other hand, in coding theory, hard problems can be considerably simplified if terminal nodes share a small amount of common randomness. For instance, the capacity of adversarial bit-flip channel is still unknown in general, but it can be characterized if shared randomness is available. In fact, Langberg~\cite{langberg2004private} shows that $\Theta(\log n)$ bits shared secrets are sufficient to force such a powerful adversary to become ``random noise". 

Motivated by the power of shared secrets, this paper focuses on node-based adversary problems with a small amount of shared secrets. Under such settings, we provide a family of network codes that are computationally efficient and information-theoretically optimal. The shared secrets between source and \emph{each} other node can either be pre-allocated  or distributed by applying network secrets sharing schemes, for example, see, for example, \cite{shah2013secure}. In addition, our network codes can be distributedly implemented and work well even when the network topology is unknown.\footnote{Our code design in Section~\ref{sec:codeconstruction} requires no knowledge of network topology.}

The rest of the paper is organized as follows. We present a concrete example to describe previous adversary models and our model of this paper in Section~\ref{sec:example}. After introducing the general network model in Section~\ref{sec:model}, we present our main results in Section~\ref{sec:mainresults}. The details of the code construction and complexity analysis are in Section~\ref{sec:codeconstruction}. Finally, we briefly discuss the generalizations and conclude in Section~\ref{sec:conclusion}.

\section{An Illustrative Example}
\label{sec:example}
Consider the network in Figure~\ref{fig:toy example network}. The source Alice wishes to transmits her message $M$ to the destination Bob through an adversarial network. The adversary Calvin, hidden somewhere in the network, can control a subset of the network and tries to corrupt the communication from Alice to Bob. In addition, through the whole paper, we assume the adversary is omniscient but casual, that is, Calvin can observe all information transmitted on the links causally. Furthermore, we focus on computationally unbounded adversary (although it is also interesting to consider the computationally bounded adversary). The goal is to find the maximum communication rate in the presence of an adversary. Before studying the capacity of this network in our setting, we examine the capacity when there is no shared secret between nodes, both for the link-based adversary model and the node-based adversary model.

\subsection{Previous Work}
\textit{Link-based adversary:} Calvin can choose \textit{any two links} to attack. Although Alice and Bob know at most two links can be controlled, they have no idea a priori which two links that Calvin chooses. The maximum achievable rate of this model is shown to be $C-2z$ by~\cite{jaggi2007resilient} (here $C$ is the minimum min-cut from Alice to Bob in the network and  $z$ is the number of \textit{links} the adversary can attack). Therefore, the capacity is $0$ with such an adversary in this example. 

The following ``symmetrization" argument shows why reliable communication is impossible in this particular setting, \emph{i.e.} $C = 3$ and $z = 2$. For any message $M$ and any code used, suppose $x_a,x_b,x_c$ are the packets induced by $M$ and the network codes on links $(a,t),(b,t),(c,t)$, respectively. Knowing the message and network codes, Calvin can choose another message $M'\neq M$ and obtain corresponding packets $x_a',x_b',x_c'$ if $M'$ were transmitted. The adversary then replaces $x_b,x_c$ by $x_b',x_c'$ on links $(b,t),(c,t)$, respectively. After receiving $x_a,x_b',x_c'$, Bob cannot distinguish between the following two events: (a) $M$ is transmitted and links $b,c$ are attacked; (b) $M'$ is transmitted and link $a$ is attacked. Therefore, no communication is possible.     

\begin{figure}[!htbp]
\centering
\includegraphics[scale=0.5]{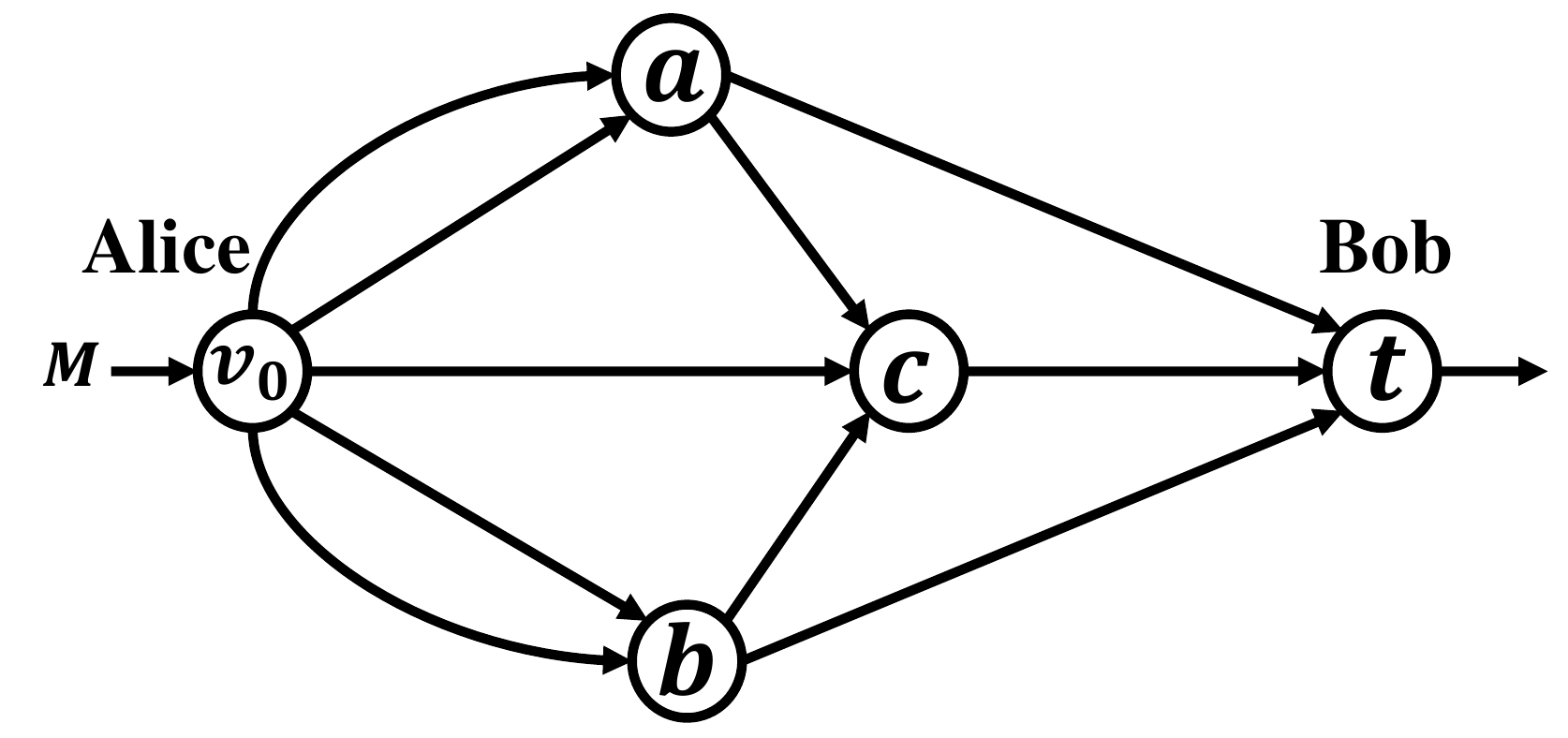}
\caption{An illustrative network example: source node $v_0$ (Alice) wishes to communicate to sink node $t$ (Bob), while adversary Calvin tries to corrupt the communication.}
\label{fig:toy example network}
\end{figure}

\textit{Node-based adversary:} Calvin can choose \textit{any one node} and attack on the outgoing links of his chosen node. Similarly, Alice and Bob know at most one node can be controlled but have no idea about which node. Since the out degree of any intermediate node is at most $2$, it is tempting to reduce such an adversary to previous link-based one,  and we get zero rate. However, it turns out that the capacity is 1, achievable by the following ``majority decoding" based code. Alice sends message $M$ on all outgoing links. All intermediate nodes perform majority decoding and forward the decoded message. A simple case-by-case analysis indicates that node $t$ can always decode $M$ correctly: If node $a$ is controlled, after majority decoding, nodes $b, c$ forward $M$ to node $t$. If node $c$ is controlled, then nodes $a,b$ forward $M$ to node $t$. The converse follows from the Singleton bound~\cite{Singleton:64} (see also~\cite[Theorem 1]{kosut2014polytope}).

\subsection{Our Model and Scheme Sketch}
In our model, in addition to treating the adversary as node-based (that can control \textit{any one node}), we also allow the source Alice to share independent common randomness (or shared secrets) with \textit{every} other node in the network. As we will see later, this is a crucial assumption that distinguishes our model from prior models. The shared secrets between source Alice and any other node, say node $a$, are sequences of bits only known to Alice and node $a$, that is, Calvin cannot access these bits unless he chooses to control node $a$. 

At a very superficial level, the role of the shared secret is to let a node \emph{just} downstream from the adversary detect any corrupted packets by verifying (using the shared secret) whether or not the received packets belong to the subspace spanned by original packets. We show that even with an asymptotically negligible rate of shared secret, rate $2$ is achievable in this example. Further, this is the best one can hope for as the adversary can always send zeros on one link in a min cut. 

In the remainder of this section, we describe the sketch of the achievability scheme for the above example. The detailed description for  general case is in Section~\ref{sec:codeconstruction}.
  
The code consists of a \emph{source encoder}, \emph{intermediate node encoders} and a \emph{destination decoder}. 
 
{\textbf{Source encoder}:} Let $r = 2$. There are two steps. First, as in random linear network coding~\cite{ho2006random}, each message $M\in\mathbb{F}_q^{nr}$ is encoded as a matrix $X\in\mathbb{F}_q^{r\times (n+r)}$ consisting of the information part $U\in\mathbb{F}_q^{r\times n}$ and coefficient header part $I\in\mathbb{F}_q^{r\times r}$ (identity matrix). The second step is crucial, which is computation of the hash header $h$. The two resulting vectors $(X_1,h)$ and $(X_2,h)$ are referred as original packets. Then random linear combinations $(a_1X_1+a_2X_2,h)$ are sent on outgoing links, where $a_1,a_2\in\mathbb{F}_q$.

In this example, $h$ consists of four parts $h_a,h_b,h_c,h_t$, each being a vector from $\mathbb{F}_q^{r^2}$ and corresponding to nodes $a,b,c,t$, respectively. The length of $h$ therefore is $\delta = 16$, which is independent of $n$. In the following, we formally describe the computation of $h_c$; all other hashes, {\em i.e.}, $h_a, h_b, h_t$ can be computed in a similar way.

Denote the shared secret $\mathbb{S}_c$ between source $v_0$ and node $c$ by $\mathbb{S}_c = (
s_{c,1}, s_{c,2}, d_{c,1,1},d_{c,1,2},d_{c,2,1},d_{c,2,1} )\in\mathbb{F}_q^{r^2+r} $, 
and rewrite $ h_c\in\mathbb{F}_q^{r^2}$ as $
h_c = ( h_{c,1,1},h_{c,1,2},h_{c,2,1},h_{c,2,2})$.
Then $h_c$ can be computed by the following so-called \textit{linearized polynomials} based on $X$ and $\mathbb{S}_c$: 
\[
h_{c,i,j} := d_{c,i,j} - \sum_{k=1}^{n+r}x_{ik}s_{c,j}^{p^k},\hspace{5mm}1\leq i,j\leq 2, 
\]
where $x_{ik}$ is the $(i,k)$-th entry in $X$. Some nice properties of the polynomials here allow the detection of corrupted packets, as we will see later. 

{\textbf{Verify-and-encode at intermediate nodes:}} Each intermediate node receives a collection of random linear combinations of original packets $X_1,X_2$ with hash headers. Every incoming packet to an intermediate node is verified using the shared secret and hash meant for that node. Depending on the outcome of the verification, each incoming packet is classified as valid if it is verified, and invalid otherwise. After this, the intermediate node sends random linear combinations of valid packets on all outgoing links. As an example, we  describe below a sample verification step at node $c$ for the packet received from node $a$. The general case is described in Section~\ref{sec:codeconstruction}. 

Suppose node $c$ receives $(W,h')$ from node $a$ and the coefficient header part in $W$ is $(1,2)$, then if $(W,h')$ is uncorrupted, $W = X_1 + 2X_2$ should hold. To verify, node $c$ computes $Q_1$ and $Q_2$ by 
\[ 
Q_1 := \sum_{k=1}^{n+r} w_k(s_{c,1} + s_{c,2})^{p^k}, 
\] 

\[
Q_2 :=  (d_{c,1,1} - h'_{c,1,1}) + (d_{c,1,2} - h'_{c,1,2}) + 2(d_{c,2,1} - h'_{c,2,1}) + 2(d_{c,2,2} - h'_{c,2,2}),
\]
where $W = (w_1,\ldots,w_{n+r})$, $h' = (h_a',h_b',h_c',h_t')$ and $h_c' = ( h_{c,1,1}', h_{c,1,2}', h_{c,2,1}',h_{c,2,2}')$. 
Packet $W$ is regarded as valid if $Q_1 = Q_2$, and invalid if $Q_1\neq Q_2$. We prove in Appendix~\ref{app:lemma} that if $W\neq X_1 + 2X_2$, then  $Q_1\neq Q_2$ w.h.p. Details are not important at present and the point here is to sketch the verification process. 

{\textbf{Decoder:}} Let $(Y_1,h_1)$, $(Y_2,h_2)$ and $(Y_3,h_3)$ be the packets received on links $(a,t)$, $(b,t)$ and $(c,t)$. Bob first verifies $Y_i$'s by previously described verification process. Since at most one node is controlled, at least two of $Y_1,Y_2,Y_3$ are valid, say $Y_1,Y_2$. Then decoding is done by solving the system of linear equations 
\[
\begin{pmatrix}
Y_1\\
Y_2
\end{pmatrix} = T
\begin{pmatrix}
X_1\\
X_2
\end{pmatrix},\notag
\]
where $T$ is the network transformation and can be obtained directly from the coefficients header part. As shown in~\cite{ho2006random}, $T$ is full rank w.h.p., therefore, rate 2 is achieved.

\section{Definitions and Network Model} 
\label{sec:model}
\subsection{Definitions}
\emph{Notation:} For a positive integer $k$,  $[k]$ denotes the set $\lbrace 1,2,\ldots,k\rbrace$. Whenever there is little scope for ambiguity, we simply use $F$ to denote the finite field $\mathbb{F}_q$, where $q = p^m$ for some prime $p$ and integer $m$.

A \emph{directed graph} is given by ${\cal G}=({\cal V},{\cal E})$, where ${\cal V}$ is the \emph{vertex set} and ${\cal E}\subset{\cal V}\times{\cal V}$ is the \emph{edge set}.  For an edge $e:=(i,j)$, we say $\text{tail}(e)=i$ and $\text{head}(e)=j$. For a node $i\in{\cal G}$, denote the set of \emph{incoming edges} as ${\cal E}_{\text{in}}(i): = \lbrace e\in{\cal E}|\text{head}(e)  = i\rbrace$, the set of \emph{outgoing edges} as ${\cal E}_{\text{out}}: = \lbrace e\in{\cal E}|\text{tail}(e)  = i\rbrace$, the set of \emph{upstream neighbours} as ${\cal N}_{\text{in}}(i):=\lbrace j\in{\cal V}|(j,i) \in {\cal E}_{\text{in}}(i)\rbrace$, and the set of \emph{downstream neighbours} as ${\cal N}_{\text{out}}(i):=\lbrace j\in{\cal V}|(i,j) \in {\cal E}_{\text{out}}(i)\rbrace$. For any subset of edges $A\subset{\cal E}$, the subgraph ${\cal G}_A$ is the graph obtained by deleting all edges in $A$; i.e. ${\cal G}_A=({\cal V},{\cal E}\setminus A)$. For a subset of nodes $Z\subset{\cal V}$, we similarly define the subgraph obtained by deleting nodes $Z$ by ${\cal G}_Z$. 

A \emph{cut} is a subset of nodes $B\subset{\cal V}$. Given subgraph ${\cal G}'=({\cal V}',{\cal E}')$, the \emph{cut-set} of cut $B$ on ${\cal G}'$ is given by
\[
\text{cut-set}(B;{\cal G}'):=\lbrace e\in{\cal E}'|e = (a,b), a\in B\cap{\cal V}', b\in {\cal V}'\setminus B\rbrace.
\] 
Given two nodes $v_1,v_2$, the \emph{minimum cut} from $v_1$ to $v_2$ on the subgraph ${\cal G}'$ is given by 
\[
\text{min-cut}(v_1,v_2;{\cal G}'):= \min_{B:v_1\in B, v_2\notin B} |\text{cut-set}(B;{\cal G}')|.
\] 
For example, in Figure~\ref{fig:toy example network}, we have $\text{min-cut}(v_0,t;{\cal G}_a) = 2$.

\subsection{Network Model}

\textbf{Network:} We model the network with a directed graph ${\cal G}=({\cal V},{\cal E})$, where ${\cal V}$ is the set of nodes (routers) and $\cal E$ is the set of links between nodes. We assume that all links have unit capacity\footnote{We discuss networks with unequal link capacities in Section~\ref{sec:conclusion}}, meaning that they can carry one symbol from $F$ per time step. We also allow multiple edges connecting the same pair of nodes. The source node $v_0$ (Alice) wishes to multicast the message $M$, which is a vector chosen uniformly at random form $F^{nr}$, to a set of destination nodes ${\cal D}:=\lbrace t_1,\ldots,t_K\rbrace$. The network model is illustrated in Figure~\ref{NetworkModel}. 

\begin{figure}[!htbp]
\centering
\includegraphics[scale = 0.5]{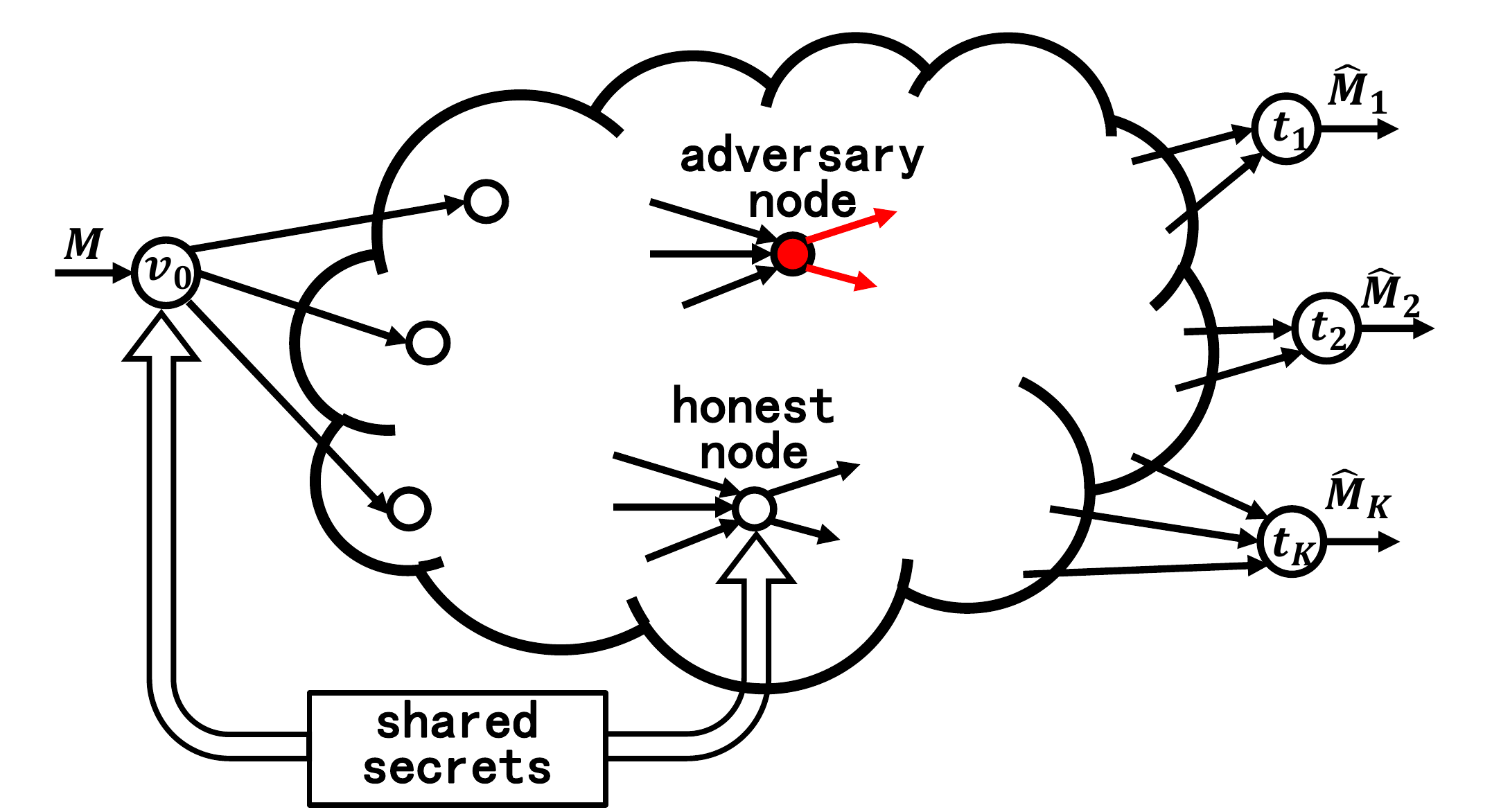}
\caption{General network model: Source $v_0$ wishes to multicast message $M$ to a set of destinations $t_{i}$'s. Adversary can selects any $z$ nodes and attack on outgoing links of these $z$ nodes. Source $v_0$ shared common randomness of negligible rate with every other node in the network. \label{NetworkModel}}
\end{figure}

\textbf{Shared secrets:} For each node $v$ other than Alice, there is an $s$-length vector $\mathbb{S}_v$, drawn uniformly at random from $F^s$, known only to Alice and node $v$.

\textbf{Adversary model: }We consider a node-based adversary (Calvin). That is, Calvin can control any $z$ nodes from ${\cal V}\backslash \lbrace v_0,t_1,\ldots,t_K\rbrace$ and transmit any information in the outgoing edges of these $z$ nodes. In other words, let ${\cal A}$ be the collection of all sets of outgoing edges of any $z$ nodes in ${\cal V}\backslash \lbrace v_0,t_1,\ldots,t_K\rbrace$; the adversary can choose any one set $A\in {\cal A}$ and inject arbitrary information on links in $A$.

The adversary is \emph{omniscient}. That is, Calvin knows the message, the code, and all packets transmitted in the network. For the $z$ nodes that Calvin selects to control, he also knows their shared secrets with source node. Calvin does not know shared secrets between the source and honest nodes (nodes that Calvin cannot control).

The table in Figure~\ref{fig:wkww} summarizes the knowledge of each communication party.

\begin{figure}[!htbp]
\centering
\includegraphics[scale=0.5]{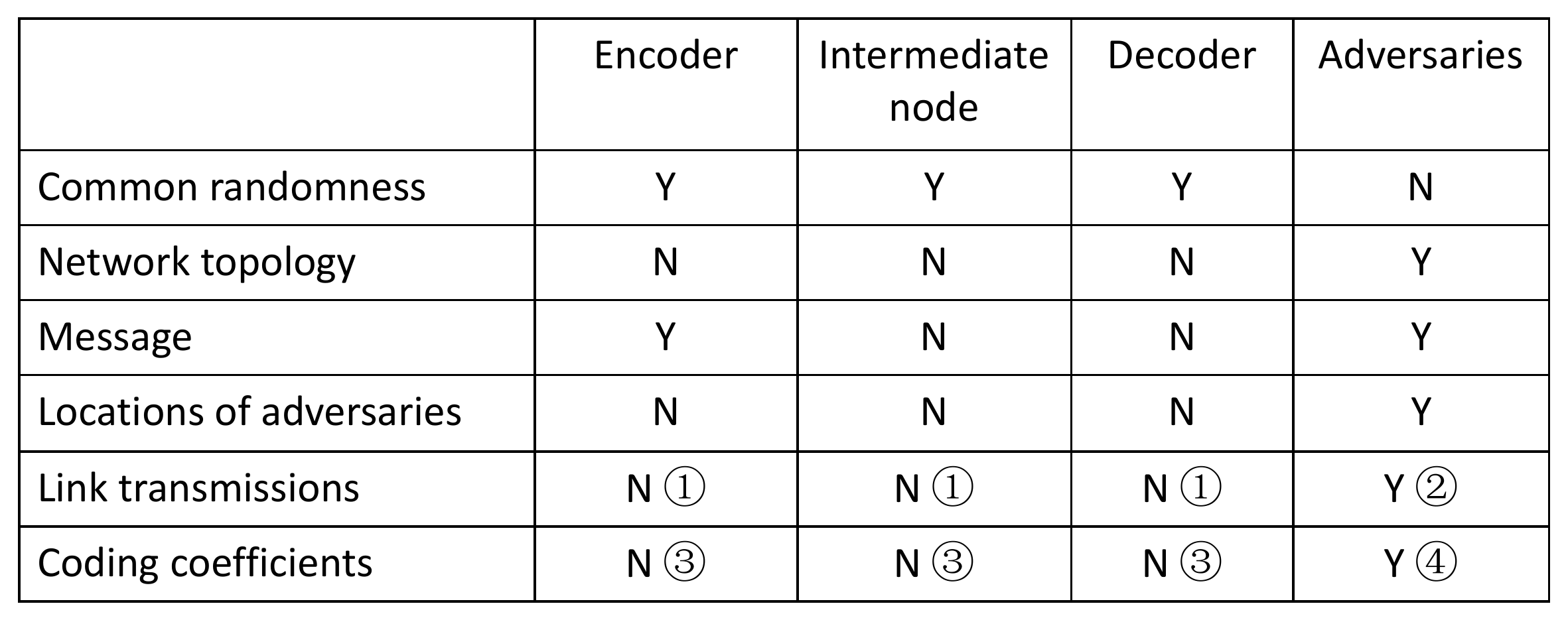}
\caption{Summary of ``who knows what when": ``Y" represents yes and ``N" represents no. Adversaries do not know the common randomness between source and honest nodes. Source, intermediate nodes and destination have no knowledge of the network topology, while the adversaries know the network topology completely. We also assume the adversaries know the transmitted message and the coding coefficients/link transmissions non-causally (\textcircled{2}\textcircled{4}). By \textcircled{1} and \textcircled{3}, we mean that the encoder, intermediate nodes and decoder only have partial knowledge on link transmissions and coding coefficients. The chronological order of the left most column describes the following relations: common randomness, network topology and message are independent with each other. The adversaries then based on the knowledge to determine the adversarial locations. Link transmissions and coding coefficients depend on common randomness, network topology, message and adversaries' locations.}
\label{fig:wkww}
\end{figure}

\textbf{Code:} A code consists of the following:
\begin{itemize}
\item \emph{Link encoders:} For each node $v$ and each link $e\in{\cal E}_{\text{out}}(v)$, a function that gives the symbol to send on $e$, given all information available to node $v$: shared secrets for node $v$, all packets from edges in ${\cal E}_{\text{in}}(v)$, and, if $v$ is Alice, the message, and shared secrets for all nodes.
\item \emph{Decoder:} For each destination node $t_k$ and each message $M_i$ for $i\in[q^{nr}]$, a function for estimating $M_i$ based on information available at node $t_k$: shared secrets and received packets.
\end{itemize}

\textbf{Code metrics:} We evaluate codes based on the following quantities:
\begin{itemize}
\item \emph{Field size} $q$,

\item \emph{Shared secret dimension} $s$,

\item \emph{Number of message symbols} $q^{nr}$,

\item \emph{Probability of error}: for any $A\in\mathcal{A}$, let $P_e(A)$ be the probability that $\hat{M}_{t_k,i}\ne M_i$ for any $t_k\in\cal D$ and any $i\in[q^{nr}]$, maximized over all possible data injections on edges in $A$, 

\item \emph{Total blocklength (bits)}, denoted by $N$,

\item \emph{Rate} $R=(nr\log q)/N$, where the base of the logarithm is 2 throughout the paper, 

\item \emph{Complexity} $\mathcal{T}$, including encoding/decoding complexity.

\end{itemize}

\section{Main Results}
\label{sec:mainresults}
We begin by stating our main results: computationally efficient achievability. Generally speaking, the natural ``erasure outer bound" can be achieved with asymptotically negligible shared secrets. In addition, computationally efficient such codes can be constructed as described in subsequent sections.   

\begin{theorem}\label{thm:no_feedback}
For any $\epsilon>0$, there exists a code satisfying
\begin{itemize}
\item $R\ge \min_{A\in{\cal A}}\min_{t_k\in\cal D}\textup{min-cut}(v_0,t_k;{\cal G}_A)-\epsilon$,
\item $P_e(A)\le \epsilon$ for all $A\in{\cal A}$,
\item $s\le \epsilon N$,
\item $N = \Theta(1/\epsilon^2)$,
\item $\mathcal{T}$ is at most polynomial in $N$ (or equivalently, in $m,n$) and network parameters $|{\cal V}|$ and $d_{\text{in}}$, as shown in Table~\ref{table:cc}. 
\end{itemize}

\end{theorem}

\begin{table}
\begin{center}
\begin{tabular}{|l|l|l|}
\hline
 & Complexity in terms of $m,n$ & Complexity in terms of $N$\\ 
\hline
Source & ${\cal O}(|{\cal V}|r^2nm \log m))$  & ${\cal O}(|{\cal V}|r^2N\log N))$ \\
Internal node & $\mathcal{O}(d_{\text{in}}nm\log m)$ & $\mathcal{O}(d_{\text{in}}N\log N)$ \\
Decoder & $\mathcal{O}((d_{\text{in}}+r^2)nm\log m)$ & $\mathcal{O}((d_{\text{in}}+r^2)N\log N)$ \\
\hline
\end{tabular}
\caption[Table caption text]{Computational complexity}
\label{table:cc}
\end{center}
\end{table}

For comparison, we also state two converse results that follow fairly directly from results already in the literature. The first states that the rate can be no larger than if the adversarial edges were simply deleted, and this ``erasure outer bound" follows from \cite{ahlswede2000network} when the residual graph is considered. This confirms that the rate achieved in Theorem~\ref{thm:no_feedback} is essentially as large as possible. The second result states that shared secrets are necessary to achieve vanishing probability of error for all adversary sets. Proofs are in the appendices.

\begin{theorem}\label{thm:converse1}
For any code if the adversary controls links $A$, the rate is upper bounded by
\[
R\le\frac{1}{1-P_e(A)}\left[ \min_{t_k\in\cal D} \textup{min-cut}(v_0,t_k;{\cal G}_A)+\frac{H(P_e(A))}{N}\right]
\]
where $H(\cdot)$ is the binary entropy function.
\end{theorem}

\begin{theorem}\label{thm:converse2}
For any code with $s=0$ (i.e. no shared secrets) and $r>0$, if there exist two adversary sets $A_1$ and $A_2$ such that $A_1\cup A_2$ covers a cut between Alice and any destination node $t_k$, then $P_e(A_1)+P_e(A_2)\ge 1$. In particular, it is impossible for both $P_e(A_1)$ and $P_e(A_2)$ to be arbitrarily small.
\end{theorem}

\section{Coding with shared secrets}
\label{sec:codeconstruction}
In this section, we state our scheme for the general case: Alice wishes to multicast her message to $K$ receivers, and any $z$ intermediate nodes can any controlled by the adversary. Let ${\cal A}$ be the collection of sets of outgoing edges of any $z$ nodes in ${\cal V}\backslash \lbrace v_0,t_1,\ldots,t_K\rbrace$. In the following, $r:= \min_{A\in{\cal A}}\min_{t_k\in\mathcal{D}}\textup{min-cut}(v_0,t_k;{\cal G}_A)$. 

\begin{mydef} A single-variable linearized polynomial (SLP) $L(x)$ over finite field $F$ is of the form $ L(x): = \sum_{i = 1}^{n}a_i x^{p^i}$, where $a_i\in F$ for any $1\leq i\leq n$ and integer $n$.
\end{mydef}

\begin{mydef}[SLP hash function]
Let $k$ be a positive integer. For  a vector ${\bf x} \in F^{k}$ and $s_1,s_2\in F$, define the SLP hash function $\psi: F^{k} \times F^2 \rightarrow  F$ by $ \psi({\bf x},s_1,s_2) := s_{2} - \sum_{l = 1}^{k} x_{l}s_{1}^{p^l}$. 
\end{mydef}

The code consists of a \emph{source encoder}, \emph{intermediate node encoders} and a \emph{destination decoders}.

\textbf{Source encoder:} Algorithm~\ref{sourceencoder} describes the encoding process, which consists of two steps. Firstly, each message $M\in [q^{nr}]$ is encoded into a matrix $X$ as described in Section~\ref{sec:example}. Denote $X_i$ the $i$-th row of $X$ and $x_{il}$ the $(i,l)$-th entry of $X$. Secondly, the hash header $h\in F^{\delta}$ is computed and appended to all $X_i$'s. Notice here that $\delta$ (specified later) is negligible as $N$ tends to infinity. Figure \ref{packet} shows the structure of the $i$-th packet. 

\begin{figure}[htbp!]
\centering
\includegraphics[scale = 0.4]{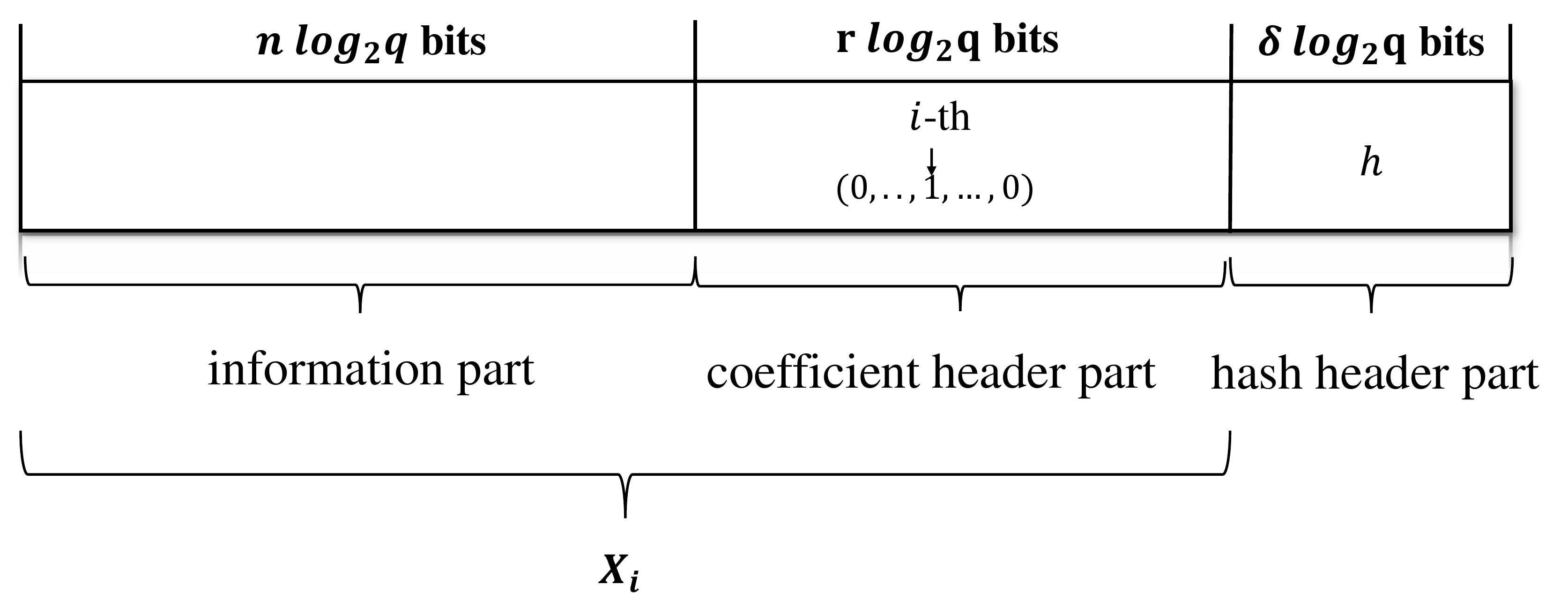}
\caption{The $i$-th packet $(X_i,h)$ and size of each part. The total length of $(X_i,h)$ is $n+r+\delta$ many finite field elements, or equivalently $(n+r+\delta)\log_2 q$ bits, where $\delta $ is independent of $n$.
\label{packet}}
\end{figure}

The header $h$ consists of $|{\cal V}|-1$ parts, with each part being a vector $h_v\in F^{r^2}$ and corresponding to a non-source node $v$. Denote entries of $h_v$ by $h_{v,i,j}$ with $1\leq i,j\leq r$. Let $\mathbb{S}_v\in F^{r^2+r}$ be the shared secrets between node $v$ and the source. These dimensions are determined by the particular hash function based on SLP's, as described below. Rewriting $\mathbb{S}_v := \lbrace s_{v,j}, d_{v,i,j}\rbrace _{1\leq i,j\leq r}$, we then compute $h_v$ by 

\[
h_{v,i,j} := \psi(X_i, s_{v,j},d_{v,i,j}) 
		  = d_{v,i,j} - \sum_{l = 1}^{n+r}x_{il}s_{v,j}^{p^l}.  
\]

A more concise form to describe the computation of $h_v$ is $ [h_v] = [d_v] - X\times [s_v]$, where $[h_v]$ denotes the matrix with entries being $h_{v,i,j}$, $[d_{v}]$ the matrix with entries being $d_{v,i,j}$ and $[s_{v}]$ the matrix with $(i,j)$-th entry being $s_{v,j}^{p^i}$, $1\leq i\leq n+r$ and $1\leq j\leq r$, i.e.,

\[
[s_{v}]: = \begin{pmatrix}
s_{v,1}^p     & s_{v,2}^p     & \dots  & s_{v,r}^p  \\
s_{v,1}^{p^2} & s_{v,2}^{p^2} & \dots  & s_{v,r}^{p^2} \\
\vdots          & \vdots          & \ddots &\vdots \\
s_{v,1}^{p^{(n+r)}} & s_{v,2}^{p^{(n+r)}} & \dots  & s_{v,r}^{p^{(n+r)}} 
\end{pmatrix}.
\]  

The hash header $h$ consists of $h_v$'s and is a row vector of length $\delta $. Each message $M$ is encoded into $r$ \emph{original packets}: $(X_i,h), 1\leq i\leq r $. Random linear combinations $(\sum_{i=1}^{r}a_iX_i,h)$ are sent on links in ${\cal E}_{\text{out}}(v_0)$, where $a_i\in F $.

\begin{algorithm}
\caption{Source encoder}\label{sourceencoder}
\begin{algorithmic}[1]
\Procedure{Source Encoder}{$M, \mathbb{S}_v, \forall v\in {\cal V}\backslash \lbrace v_0,t\rbrace$}\\
Rewrite $M$ as $X$; 
\ForAll{$v\in {\cal V}\backslash \lbrace v_0,t\rbrace$} 
\State Compute $[s_v]$;
\State Compute $[h_v] := [d_v] - X\times [s_v]$; 
\EndFor
\State Obtain $h$ by rewriting all $h_v$'s in a row vector; 
\ForAll {$e = (v_0,v)\in\mathcal{E}_{\text{out}}(v_0)$} 
\State  Create and send $(\sum_{i = 1}^{r} a_{e,i}X_i, h)$ on $e$; where $a_{e,i}$'s are generated uniformly at random from $F$.
\EndFor
\EndProcedure
\end{algorithmic}
\end{algorithm}

\emph{Encoding complexity analysis:} First of all, notice that multiplication in $\mathbb{F}_{p^m}$ can be done in ${\cal O}(m\log m)$ time, see \cite{Gathen:2003:MCA:945759}. For each node $v$, 
\begin{itemize}
\item Computation of $[s_v]$ costs ${\cal O}(rnm \log m))$ time: For each $a\in\mathbb{F}_{p^m}$, computation of $a^p$ costs $\log p$ many multiplications in $\mathbb{F}_{p^m}$, which in total is $\mathcal{O}((\log p)m\log m )$ (which again is $\mathcal{O}(m\log m )$ since $p$ is a fixed constant). Therefore, computation of $[s_v]$ costs $\mathcal{O}(r(n+r)m\log m)$ time.

\item Computation of $[h_v]$ costs ${\cal O}(r^2nm \log m))$ time: The multiplication $X\times [s_v]$ costs most, which is in ${\cal O}(r^2(n+r)m\log m)$ time.
\end{itemize}

In total, the encoding complexity is ${\cal O}(|{\cal V}|r^2nm \log m)$. Since our scheme requires $m = \Theta(n)$ to ensure small error probability (see Appendix~\ref{app:lemma}), equivalently, the encoding complexity is $\mathcal{O}(|\mathcal{V}|r^2N\log N)$, where $N$ is the number of bits in each packet.

\textbf{Verify-and-encode at intermediate nodes:} Each intermediate node $u$ performs a  ``verify-and-encode" procedure. 

For each packet $(W(e),h(e))$\footnote{Here we use $h(e)$ to denote the hash header received on link $e$ to indicate that the adversary may also corrupt the hash header $h$.} received from edge $e\in\mathcal{E}_{\text{in}}(u)$, node $u$ first \emph{verifies} whether $W(e)$ is a correct linear combination of $X_1,\ldots, X_r$, based on the hash $h(e)$ and shared secrets $\mathbb{S}_u $. Then node $u$ forwards a random linear combination of all ``valid" packets on links in $\mathcal{E}_{\text{out}}(u)$. We now describe the verification process (Algorithm~\ref{internalencoder}).  

Denote $W(e) = (w_1,\ldots,w_n, w_{n+1},\ldots,w_{n+r})$, then we know that $W(e)$ should be the linear combination $\sum_{i=1}^{r}w_{n+i}X_i$ if $W(e)$ is not corrupted by the adversary. Node $u$ computes $Q_1,Q_2$ defined below. Then $W(e)$ is regarded as valid if $Q_1=Q_2$ and invalid otherwise.

\begin{align*}
Q_1 &:= \sum_{l=1}^{n+r}w_{l}\Big(\sum_{i' = 1}^{r}\mathbbm{1}_{i'}\times s_{u,i'}\Big)^{p^l}\\
Q_2 &:=  \sum_{i=1}^{r}w_{n+i}\sum_{i' = 1}^{r}\mathbbm{1}_{i'}\times (d_{u,i,i'} - h(e)_{u,i,i'})
\end{align*}

Here $\mathbbm{1}_{i'}$ is the indicator function of event $\lbrace w_{n+i'}\neq 0\rbrace$. In Section~\ref{sec:example}, we have described the verification procedure in details for the specific packet $W = X_1 + 2X_2$. This scheme works w.h.p. over the randomness in the shared secrets. The following lemma states this formally.

\begin{lemma}\label{lemma1} When $W(e)\neq\sum_{i=1}^{r}w_{n+i}X_i$, then $Q_1\neq Q_2$ with probability at least $1-1/p^{\Theta(n)}$ when the field size is $q = p^{\Theta(n)}$. When $W(e)=\sum_{i=1}^{r}w_{n+i}X_i$ and $[h'_u] = [h_u]$, then $Q_1 = Q_2$ always holds.
\end{lemma} 

\begin{proof}
See Appendix~\ref{app:lemma}. As a remark, the relation between the block length $N$ (in bits) and $m,n$ is $N = \Theta(nm)$. Therefore, the general relations are: $m = \Theta(\sqrt{N})$, $ n = \Theta(\sqrt{N})$, $q = p^m$, and  $ m = cn $ for some constant $c>1$. To be concrete, we can think of the following parameter settings: $n = \sqrt{N/2}$, $m = 2n = \sqrt{2N}$ and $q = 2^{\sqrt{2N}}$. With this setting, the probability that a corrupted packet is not detected with probability at most $2^{-n} = 2^{-\sqrt{N/2}}$. 
\end{proof}

\begin{algorithm}
\caption{Internal encoder}\label{internalencoder}
\begin{algorithmic}[1]
\Procedure{Verifier}{$ \mathbb{S}_u$} 
\ForAll{$e\in {\cal E}_{\text{in}}(u)$} 
\State Compute $Q_1,Q_2$;
\State if $Q_1=Q_2$, then label $W(e)$ as valid;
\State else, label $W(e)$ as invalid.
\EndFor
\EndProcedure
\Procedure{ Encoder}{all valid $W(e)$, $e\in\mathcal{E}_{\text{in}}(u)$} 
\ForAll {$e\in\mathcal{E}_{\text{out}}(u)$}
\State Create and send $(\sum_{e'\in\mathcal{E}_{\text{in}}(u)} a_{e'}W(e'), h)$ on $e$; where $a_{e'} = 0$ if $W(e')$ is invalid; $a_{e'}$ is generated uniformly at random from $F$ if $W(e')$ is valid.
\EndFor
\EndProcedure
\end{algorithmic}
\end{algorithm} 

\emph{Complexity analysis:} For each intermediate node $u$, 
\begin{itemize}
\item Computation of $Q_1$ costs ${\cal O}(d_{\text{in}}nm \log m))$ time: For each incoming link of $u$, computation of $Q_1$ for this link costs at most $2(n+r)$ many multiplications. Thus, the complexity for computations of $Q_1$'s on all incoming links cost ${\cal O}(d_{\text{in}}(n+r)m\log m)$ time, where $d_{\text{in}}$ denotes the maximum in-degree in the graph.   

\item Computation of $Q_2$ costs ${\cal O}(d_{\text{in}}rm \log m))$ time: For each link, there are at most $2r$ additions and $r$ multiplications in $F$ needed for the computation of $Q_2$.  
\end{itemize}

In total, the complexity of each intermediate node is ${\cal O}(d_{\text{in}}nm\log m)$, or equivalently, ${\cal O}(d_{\text{in}}N\log N)$ with $N$ be the number of bits in each packet. 

\textbf{Decoder:} Each destination node verifies received packets, and decodes using all valid ones (see Algorithm~\ref{decoder}). The \emph{decoding complexity} for each decoder is $\mathcal{O}((d_{\text{in}}+r^2)nm\log m)$, or equivalently, 
$\mathcal{O}((d_{\text{in}}+r^2)N\log N)$ as analysed in the following. 

\begin{algorithm}
\caption{Destination decoder}\label{decoder}
\begin{algorithmic}[1]
\Procedure{Verifier}{$\mathbb{S}_t$} 
\ForAll{$e\in {\cal E}_{\text{in}}(t)$} 
\State Compute $Q_1,Q_2$;
\State if $Q_1=Q_2$, then label $W(e)$ as valid;
\State else, label $W(e)$ as invalid.
\EndFor
\State Rearrange all valid packets as $Y$.
\EndProcedure
\Procedure{Decoder}{all valid $W(e)$, $e\in\mathcal{E}_{\text{in}}(t)$}
\State Solve $Y = TX$ with unknowns being $X$ ;
\EndProcedure
\end{algorithmic}
\end{algorithm}

\emph{Decoding complexity analysis: } For each destination node decoder,

\begin{itemize}
\item Computations of $Q_1, Q_2$ cost ${\cal O}(d_{\text{in}}nm\log m)$ time, the same as each intermediate node;

\item Decoding is done by solving a system of linear equations $Y = TX$. Notice here that the matrix $T$ is different for each decoder in general. Dimensions of these three matrices are $X\in F^{r\times (n+r)}$, $T\in F^{r\times r}$ and $Y\in F^{r\times (n+r)}$. The complexity of Gaussian elimination is ${\cal O}(r^2nm\log m)$.   
\end{itemize} 

Therefore, the total complexity for each decoder is ${\cal O}((d_{\text{in}} + r^2)nm\log m)$, or equivalently ${\cal O}((d_{\text{in}} + r^2)N\log N)$.

\section{Discussion and Conclusion}
\label{sec:conclusion}

We develop novel computationally-efficient network codes using shared secrets for node-based adversary problems. Our codes meet a natural erasure outer bound. 

Although the description of our codes is specialized for the case of node-based adversaries in multicast settings,  our techniques also extend to more general cases. We briefly describe some of them in the following.


\begin{itemize}
\item\emph{General adversarial sets} \label{subsec:gas}

Let $\cal A$ be collection of all possible subsets of links that the adversary can control. The set $\cal A$ is given and known to both the transmitter and receivers, however, the specific set $A\in\mathcal{A}$ of attacked links is a priori known only to the adversary. Since our code is essentially designed to detect corrupted links, it works even for this general adversarial set $\cal A$. Note that link-based (respectively, node-based) adversaries correspond to $\mathcal{A}$ being the collection of all subsets of $z$ links (respectively, collection of sets of outgoing links from all subsets of $z$ nodes). In this setting, the code has similar guarantees as Theorem~\ref{thm:no_feedback}. The proof also follows on similar lines and is skipped here.
 
\begin{corollary}
For any $\epsilon>0$, there exists a code satisfying
\begin{itemize}
\item $R\ge \min_{A\in{\cal A}}\min_{t_k\in\cal D}\textup{min-cut}(v_0,t_k;{\cal G}_A)-\epsilon$,
\item $P_e(A)\le \epsilon$ for all $A\in{\cal A}$,
\item $s\le \epsilon N$,
\item $N = \Theta(1/\epsilon^2)$,
\item $\mathcal{T}$ is at most polynomial in $N$ (or equivalently, in $m,n$) and network parameters $|{\cal V}|$ and $d_{\text{in}}$, as shown in Table~\ref{table:cc}. 
\end{itemize}
\end{corollary}

\item\emph{Unequal link capacities}

Networks with unequal link capacities can be handled in a manner similar to unit capacities networks with  general adversarial set. For integer (or rational) link capacities, Figure~\ref{fig:reduction} shows a simple example where such a reduction is done. The case of irrational link capacities can be approximated by rational capacities. Therefore, this is a special case of the general adversarial set.

\item\emph{Network with cycles} 

Although our codes here cannot be directly applied to network with cycles, it is possible to generalize our ideas by working with the corresponding \emph{acyclic} time-expanded network ({\em c.f.}~\cite[Chapter~20]{yeung2008information}).

\begin{figure}[!htbp]
\centering
\includegraphics[scale=0.5]{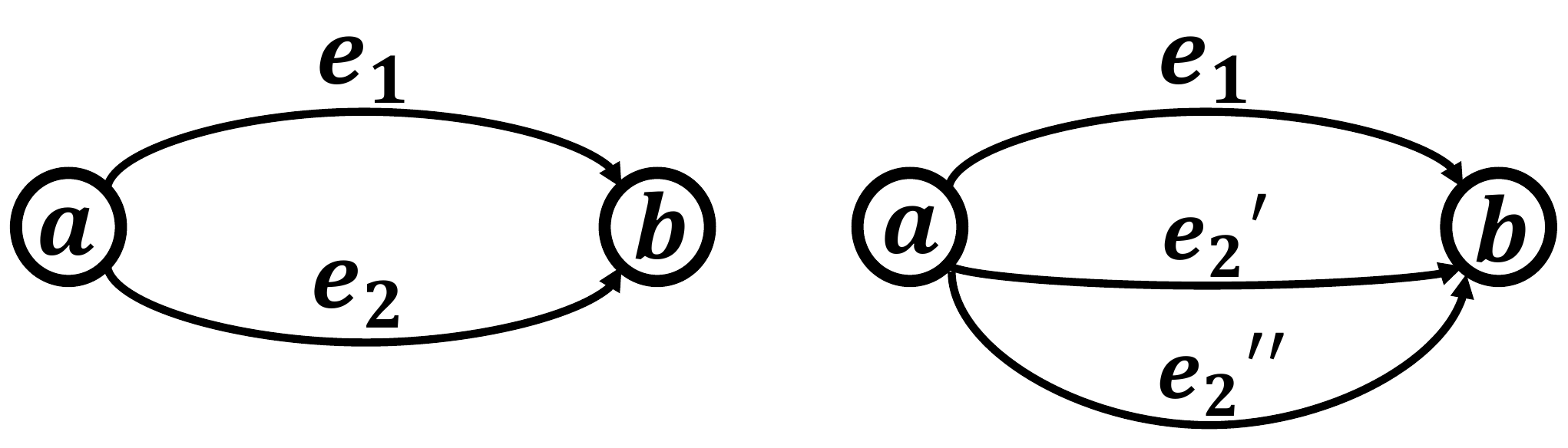}

\caption{Left: link $e_1$ has capacity 1 and link $e_2$ has capacity 2. The adversarial set ${\cal A} = \left\{\lbrace e_1 \rbrace, \lbrace e_2 \rbrace\right\}$. Right: link $e_1,e_2',e_2''$ all have capacity 1 and the adversarial set ${\cal A}' = \left\{\lbrace e_1 \rbrace, \lbrace e_2',e_2'' \rbrace\right\}$. }
\label{fig:reduction}%
\end{figure}
\end{itemize}

\appendices
\section{Proof of Lemma \ref{lemma1}}
\label{app:lemma}
\begin{proof}
Let $W'(e): = (w'_1,\ldots,w'_n,w_{n+1},\ldots,w_{n+r})$ $ = \sum_{i=1}^{r} w_{n+i}X_i $ be the correct linear combination claimed by header of $W(e)$. We bound the probability of the event $E_1 :=\lbrace  Q_1=Q_2 \rbrace$ conditioned on $E_2:=\lbrace W(e)\neq W'(e)\rbrace$. 

Rewrite $Q_1$ as the following
\begin{align*}
Q_1 &= \sum_{l = 1}^{n+r}(w'_l + (w_l - w'_l))\sum_{i' = 1}^{r}\mathbbm{1}_{i'}s_{u,i'}^{p^l}\\
    &= \sum_{i=1}^{r}w_{n+i}\sum_{i' = 1}^{r}\mathbbm{1}_{i'}\times (d_{u,i,i'} - h_{u,i,i'}) + \sum_{l = 1}^{n}(w_l - w'_l)\sum_{i' = 1}^{r}\mathbbm{1}_{i'}s_{u,i'}^{p^l}
\end{align*}

Then, $Q_1 = Q_2$ is equivalent to 
\[
\sum_{l = 1}^{n}(w_l - w'_l)\sum_{i' = 1}^{r}\mathbbm{1}_{i'}s_{u,i'}^{p^l} + \sum_{i=1}^{r}w_{n+i}\sum_{i' = 1}^{r}\mathbbm{1}_{i'}\times (h'_{u,i,i'} - h_{u,i,i'}) = 0, 
\] 
where the left hand side is a non-zero polynomial of degree at most $p^n$ in $s_{u,1},\ldots,s_{u,r}$. Since the shared secrets $s_{u,1},\ldots,s_{u,r}$ are uniformly generated from $\mathbb{F}_{p^m}$, by the Schwartz Zippel  Lemma (see~\cite{motwani2010randomized}), we have $P(E_1|E_2)\leq\frac{p^n}{p^m}$, which is exponentially small in $n$ when $m=cn$ for some constant $c>1$. The case of $W(e) = W'(e)$ follows from direct computation, which we omit here.
\end{proof}

\section{Proof of  Theorem~\ref{thm:no_feedback}}
\label{app:no_feedback}
\begin{proof}
Let $r = \min_{A\in{\cal A}}\min_{t_k\in\cal D}\textup{min-cut}(v_0,t_k;{\cal G}_A)$ and $\delta := (|{\cal V}|-1)r^2$. Denote the maximum node degree (sum of in-degree and out-degree) in $\cal{G}$ as $d_{\text{max}}$.  

Given any $\epsilon>0$, let $n = \sqrt{N/2}$, $m = 2n = \sqrt{2N}$ and $q = 2^{\sqrt{2N}}$, then take $N$ large enough such that $z\cdot d_{\text{max}}\cdot 2^{-\sqrt{N/2}} < \epsilon $ and $ \frac{nr}{n+r+\delta} > r-\epsilon $. Design network code described in Section~\ref{sec:codeconstruction} over $\mathbb{F}_{q}$. Denote the constructed network code by $\mathbb{C}(r,\epsilon)$. We analyse the code $\mathbb{C}(r,\epsilon)$ in the following:

\begin{itemize}
\item \emph{Rate:} $R = \frac{nr}{n+r + \delta} > r - \epsilon$ by choice of $n$;
\item \emph{Error probability:} For any $A\in\mathcal{A}$, the error probability is upper bounded by the probability that some corrupted packet is not detected by the very next downstream node, which, by union bound, is at most $z\cdot d_{\text{max}}\cdot 2^{-\sqrt{N/2}} < \epsilon $ by taking $N$ to be large. Therefore, $P_e(A)\leq\epsilon$ for any $A\in\mathcal{A}$.
\item \emph{Dimension of shared secrets:} Our scheme in Section~\ref{sec:codeconstruction} requires $s = r^2+r$ for each node, which is independent of $n$ (also independent of $N$).

\item \emph{Block length:} $N = \Theta(1/\epsilon^2)$ follows from the error probability requirement $z\cdot d_{\text{max}}\cdot 2^{-\sqrt{N/2}} < \epsilon $ and rate requirement $\frac{nr}{n+r+\delta} > r-\epsilon$. The former requires $N >  2\log^2(\frac{z\cdot d_{\text{max}}}{\epsilon})$ and the latter requires $N>\frac{|{\cal V}|^\alpha}{\epsilon^2}$ for some constant $\alpha$ (a crude estimation yields $\alpha = 7$). Since we are interested in the regime where $\epsilon$ tends to zero, we have $N = \Theta(1/\epsilon^2)$.

\item \emph{Complexity:} The complexity is given in our description of the code design in Section~\ref{sec:codeconstruction}. 
\end{itemize}

Note that not all corrupted packets are detectable and honest nodes can be isolated by adversarial nodes, as shown in Figure~\ref{fig:multiple}. However, these cases do not change the minimum cut and the theorem still holds.

\begin{figure}[!htbp]
\centering
\subfigure[]{
\includegraphics[scale = 0.25]{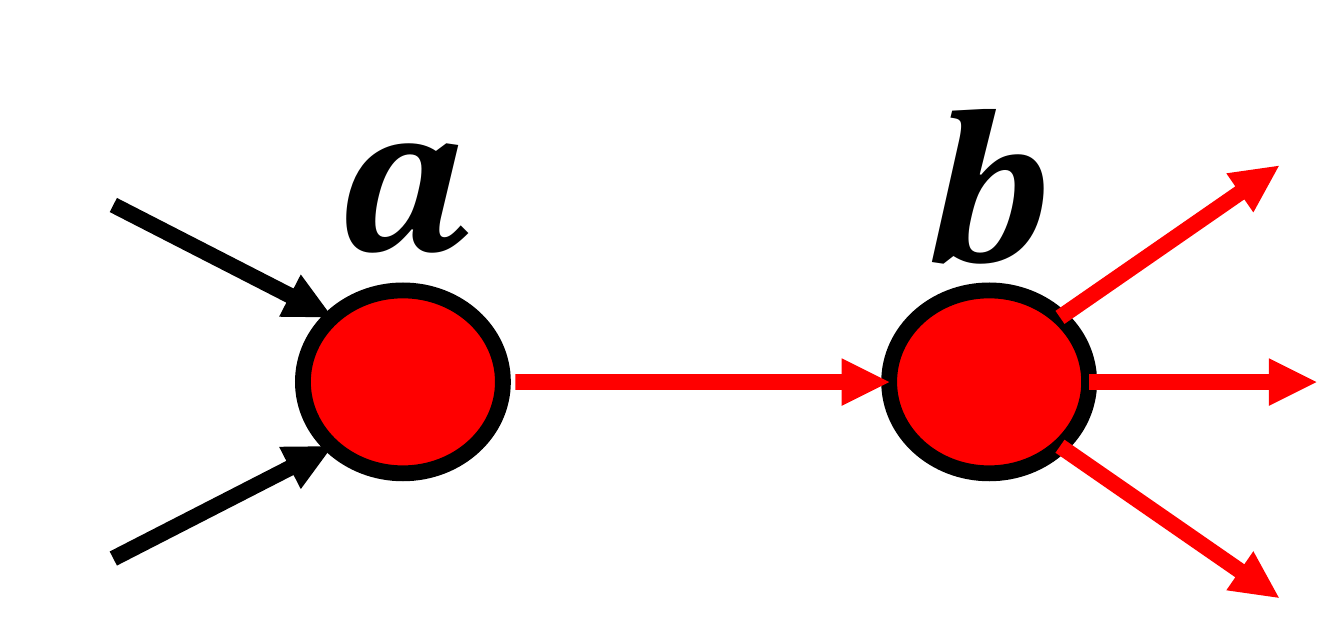}
}

\quad
\subfigure[]{
\includegraphics[scale = 0.35]{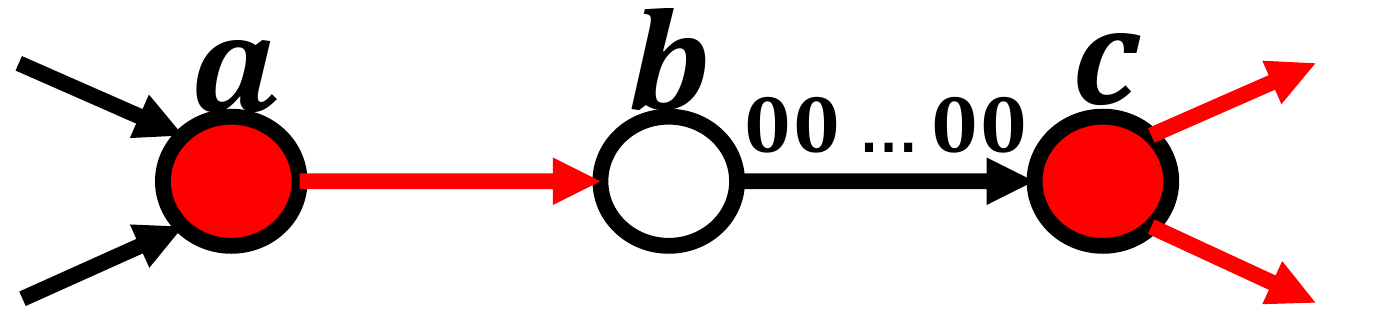}
}

\caption{(a) A simple example shows that corrupted packets may not be detectable. When nodes $a$ and $b$ are both controlled by the adversary, the corrupted packet from $a$ to $b$ is not detectable. (b) A simple example shows that honest node can be isolated by adversarial nodes. }
\label{fig:multiple}%
\end{figure}
\end{proof}

\section{Proof of Theorem~\ref{thm:converse1}}
\label{app:converse1}

\begin{proof}
Let $A\in\mathcal{A}$ and $r=\min_{t_k\in\cal D} \textup{min-cut}(v_0,t_k;{\cal G}_A)$. Suppose adversary sends zeros on links in $A$, and let $E\subset\mathcal{E}$ be the cut-set of a minimum cut between $v_0$ to $t_k$ (minimized over all $k$), then $|E|=r$. Denote $W_1,\ldots, W_{r}$ be the random variables induced by any codes on links in $E$, then we can get a Markov chain $M\rightarrow (W_1,\ldots,W_{r})\rightarrow\hat{M}$. The proof is completed by the following: 
\begin{align*}
NR &=  H(M)  =  H(M|\hat{M}) + I(M;\hat{M}) \\
   &\leq  H(M|\hat{M}) + I(M; W_1,\ldots,W_{r}) \\
   &\leq  H(P_e(A)) + P_e(A)NR + H(W_1,\ldots,W_{r}) \\
   &\leq  H(P_e(A)) + P_e(A)NR + Nr. 
\end{align*}
Therefore, $R\leq \frac{1}{1-P_e(A)}\left [r + \frac{H(P_e(A))}{N} \right ]$.
\end{proof}

\section{Proof of Theorem~\ref{thm:converse2}}
\label{app:converse2}
\begin{proof}
The proof is based on a symmetrization argument (similar to, e.g.\cite[Theorem 1]{kosut2014polytope}). Suppose there is a code $\mathbb{C}$, which has rate $R>0$. There exist two messages $w_1$ and $w_2$ such that the packets induced by this code on edges in $A_1$ and $A_2$ are $\mathbb{C}(w_1,A_1),\mathbb{C}(w_1,A_2)$ and $\mathbb{C}(w_2,A_1),\mathbb{C}(w_2,A_2)$, respectively. The adversary then adopts the following attack strategy such that the receiver cannot distinguish which one of $w_1,w_2$ is transmitted, thus causing non-vanishing probability of decoding error.
\begin{itemize}
\item If $w_1$ is transmitted, then Calvin replaces packets $\mathbb{C}(w_1,A_1)$ by $\mathbb{C}(w_2,A_1)$;
\item If $w_2$ is transmitted, then Calvin replaces packets $\mathbb{C}(w_2,A_2)$ by $\mathbb{C}(w_1,A_2)$;
\end{itemize} 
Then once the destination node receives $\mathbb{C}(w_2,A_1),\mathbb{C}(w_1,A_2)$, the decoder cannot distinguish between the above two events, and so $w_1$ and $w_2$ are not distinguishable, and hence the probability of decoding error is at least 1/2.
\end{proof}


\section{Table of Notations}
All notations used are listed in the following table.
\begin{table}[!htbp]
\begin{center}
\small
\begin{tabular}{ |l|l|l| }
  \hline
  \multicolumn{3}{|c|}{Model Parameters} \\
  \hline
  Notation & Meaning & Value/Range\\
  \hline 
  ${\cal G} = ({\cal V}, {\cal E})$  & directed acyclic network with vertex set ${\cal V}$  and edge set ${\cal E}$ 
  & - \\
  $v_0$   & source node                & $v_0\in{\cal V}$\\
  $t_i$   & the $i$-th receiver node   & $t_i\in{\cal V}$\\
  ${\cal E}_{in}(v)$  & set of incoming edges at node $v$         & $\subseteq\mathcal{E}$\\
  ${\cal E}_{out}(v)$ & set of outgoing edges from node $v$       & $\subseteq\mathcal{E}$\\
  ${\cal N}_{in}(v)$  & set of adjacent upstream nodes of $v$     & $\subseteq\mathcal{V}$\\
  ${\cal N}_{out}(v)$ & set of adjacent downstream nodes of $v$   & $\subseteq\mathcal{V}$\\
  $d_{\text{in}}$    & maximum node in-degree in ${\cal G}$       & \\
  $d_{\text{max}}$    & maximum node degree in ${\cal G}$         & \\  
  \hline
  \multicolumn{3}{|c|}{Code Parameters} \\
  \hline
  $p$  & a fixed prime number  & \\
  $q$  & field size            & $q = p^m$ for some positive integer $m$\\
  $F$  &  simplified notation for finite field        & $F: = \mathbb{F}_{p^m}$\\
  $M$  & message     & $\in [q^{nr}]$\\
  $s$  & dimension of shared secrets    & \\
  $\mathbb{S}_{u}$  & shared secrets between source $v_0$ and node $u$   & $\in F^s$\\
  $N$  & block length in bits & \\
  $n$  & number of payload (finite field elements) in each packet & \\
  $P_e$ & error probability  & \\
  $R$   & rate of network codes & $R = \frac{\log q^{nr}}{N}$\\
  \hline
\end{tabular}
\caption[Table caption text]{Notations}
\label{table: notation}
\end{center} 
\end{table}

\bibliographystyle{ieeetr}
\bibliography{manuscripts}

\begin{thebibliography}{10}

\bibitem{ahlswede2000network}
R.~Ahlswede, N.~Cai, S.-Y.~R. Li, and R.~W. Yeung, ``Network information
  flow,'' {\em IEEE Transactions on Information Theory}, vol.~46, no.~4,
  pp.~1204--1216, 2000.

\bibitem{ho2003benefits}
T.~Ho, R.~Koetter, M.~Medard, D.~R. Karger, and M.~Effros, ``The benefits of
  coding over routing in a randomized setting,'' in {\em Proceedings of 2003
  IEEE International Symposium on Information Theory}, p.~442, 2003.

\bibitem{li2003linear}
S.-Y.~R. Li, R.~W. Yeung, and N.~Cai, ``Linear network coding,'' {\em IEEE
  Transactions on Information Theory}, vol.~49, no.~2, pp.~371--381, 2003.

\bibitem{koetter2003algebraic}
R.~Koetter and M.~M{\'e}dard, ``An algebraic approach to network coding,'' {\em
  IEEE/ACM Transactions on Networking (TON)}, vol.~11, no.~5, pp.~782--795,
  2003.

\bibitem{jaggi2005polynomial}
S.~Jaggi, P.~Sanders, P.~Chou, M.~Effros, S.~Egner, K.~Jain, L.~M. Tolhuizen,
  {\em et~al.}, ``Polynomial time algorithms for multicast network code
  construction,'' {\em IEEE Transactions on Information Theory}, vol.~51,
  no.~6, pp.~1973--1982, 2005.

\bibitem{ho2006random}
T.~Ho, M.~M{\'e}dard, R.~Koetter, D.~R. Karger, M.~Effros, J.~Shi, and
  B.~Leong, ``A random linear network coding approach to multicast,'' {\em IEEE
  Transactions on Information Theory}, vol.~52, no.~10, pp.~4413--4430, 2006.

\bibitem{lun2005efficient}
D.~S. Lun, M.~M{\'e}dard, and R.~Koetter, ``Efficient operation of wireless
  packet networks using network coding,'' in {\em International Workshop on
  Convergent Technologies}, 2005.

\bibitem{yeung2006network}
R.~W. Yeung and N.~Cai, ``Network error correction, i: Basic concepts and upper
  bounds,'' in {\em Proceedings of Communications in Information and Systems
  2006}, vol.~6, pp.~19--35.

\bibitem{cai2006network}
N.~Cai and R.~W. Yeung, ``Network error correction, ii: Lower bounds,'' in {\em
  Proceedings of Communications in Information and Systems 2006}, vol.~6,
  pp.~37--54.

\bibitem{jaggi2007resilient}
S.~Jaggi, M.~Langberg, S.~Katti, T.~Ho, D.~Katabi, and M.~M{\'e}dard,
  ``Resilient network coding in the presence of byzantine adversaries,'' in
  {\em Proceedings of 26th IEEE International Conference on Computer
  Communications}, pp.~616--624, IEEE, 2007.

\bibitem{koetter2008coding}
R.~Koetter and F.~R. Kschischang, ``Coding for errors and erasures in random
  network coding,'' {\em IEEE Transactions on Information Theory}, vol.~54,
  no.~8, pp.~3579--3591, 2008.

\bibitem{yao2014network}
H.~Yao, D.~Silva, S.~Jaggi, and M.~Langberg, ``Network codes resilient to
  jamming and eavesdropping,'' {\em IEEE/ACM Transactions on Networking (TON)},
  vol.~22, no.~6, pp.~1978--1987, 2014.

\bibitem{nutman2008adversarial}
L.~Nutman and M.~Langberg, ``Adversarial models and resilient schemes for
  network coding,'' in {\em Proceedings of 2008 IEEE International Symposium on
  Information Theory}, pp.~171--175, 2008.

\bibitem{zhang2015coding}
Q.~Zhang, S.~Kadhe, M.~Bakshi, S.~Jaggi, and A.~Sprintson, ``Coding against a
  limited-view adversary: The effect of causality and feedback,'' in {\em
  Proceedings of 2015 IEEE International Symposium on Information Theory},
  pp.~2530--2534, 2015.

\bibitem{kosut2010adversaries}
O.~E. Kosut, {\em Adversaries in networks}.
\newblock PhD thesis, Cornell University, 2010.

\bibitem{kosut2014polytope}
O.~Kosut, L.~Tong, and D.~N. Tse, ``Polytope codes against adversaries in
  networks,'' {\em IEEE Transactions on Information Theory}, vol.~60, no.~6,
  pp.~3308--3344, 2014.

\bibitem{che2013routing}
P.~H. Che, M.~Chen, T.~Ho, S.~Jaggi, and M.~Langberg, ``Routing for security in
  networks with adversarial nodes,'' in {\em Proceedings of 2013 International
  Symposium on Network Coding (NetCod)}, pp.~1--6, 2013.

\bibitem{kim2011network}
S.~Kim, T.~Ho, M.~Effros, and A.~S. Avestimehr, ``Network error correction with
  unequal link capacities,'' {\em IEEE Transactions on Information Theory},
  vol.~57, no.~2, pp.~1144--1164, 2011.

\bibitem{kosut2014generalized}
O.~Kosut and L.-W. Kao, ``On generalized active attacks by causal adversaries
  in networks,'' in {\em Proceedings of 2014 IEEE Information Theory Workshop},
  pp.~247--251, 2014.

\bibitem{huang2015connecting}
W.~Huang, M.~Langberg, and J.~Kliewer, ``Connecting multiple-unicast and
  network error correction: Reduction and unachievability,'' in {\em
  Proceedings of 2015 IEEE International Symposium on Information Theory},
  pp.~361--365, 2015.

\bibitem{langberg2004private}
M.~Langberg, ``Private codes or succinct random codes that are (almost)
  perfect,'' in {\em Proceedings of 45th Annual IEEE Symposium on Foundations
  of Computer Science}, pp.~325--334, IEEE, 2004.

\bibitem{shah2013secure}
N.~B. Shah, K.~Rashmi, and K.~Ramchandran, ``Secure network coding for
  distributed secret sharing with low communication cost,'' in {\em Proceedings
  of 2013 IEEE International Symposium on Information Theory}, pp.~2404--2408,
  2013.

\bibitem{Singleton:64}
R.~C. Singleton, ``Maximum distance q-nary codes,'' {\em IEEE Transactions on
  Information Theory}, vol.~10, pp.~116--118, Apr 1964.

\bibitem{Gathen:2003:MCA:945759}
J.~V.~Z. Gathen and J.~Gerhard, {\em Modern Computer Algebra}.
\newblock New York, NY, USA: Cambridge University Press, 2~ed., 2003.

\bibitem{yeung2008information}
R.~W. Yeung, {\em Information theory and network coding}.
\newblock Springer Science \& Business Media, 2008.

\bibitem{motwani2010randomized}
R.~Motwani and P.~Raghavan, {\em Randomized algorithms}.
\newblock Chapman \& Hall/CRC, 2010.

\end{thebibliography}
\end{document}